\documentclass[12pt]{article} 
\pdfoutput=1 
\usepackage[dvips]{graphicx}
\usepackage{amssymb,amsfonts,amsmath}

\usepackage{enumerate}

\usepackage[open-square,define-standard-theorems]{QED} %

\newcommand{\yeslabel}{{\color{ForestGreen} Yes}}
\newcommand{\nolabel}{{\color{RubineRed} No}}

\newcommand\resultstable{Table~1}
\newcommand\sirefproofmainprop{Appendix~A}
\newcommand\sirefproofcomp{Appendix~B}
\newcommand\sirefproposals{Appendix~C}
\newcommand\sirefcompaspects{Appendix~D}

\usepackage{textcomp,mathrsfs}
\usepackage{graphicx,amssymb,amsmath,mathrsfs}
\usepackage{multirow,makeidx,algorithmic,algorithm}
\usepackage{latexsym,graphicx,amssymb,amsmath,mathrsfs}
\usepackage{mathrsfs}
\usepackage{amssymb}
\usepackage{booktabs,fancyhdr}
\usepackage[usenames,dvipsnames]{color}

\usepackage{url}
\begingroup
\makeatletter
\g@addto@macro{\UrlSpecials}{%
  \endlinechar=13 \catcode\endlinechar=12
  \do\%{\Url@percent}\do\^^M{\break}}
 \gdef\Url@percent{\@ifnextchar^^M{\@gobble}{\mathbin{\mathchar`\%}}}%
\endgroup %

\graphicspath{{./graphics/},{./figures/}}

\newcounter{xxx}
\DeclareMathOperator{\Unifd}{Unif}
\DeclareMathOperator{\Expd}{Exp}

\DeclareMathOperator{\Poid}{Poi}
\newcommand\phylospace{{\mathcal T}}
\newcommand\alignspace{{\mathcal M}}

\newcommand\bracearraycond[1]{\left\{ \begin{array}{ll} #1 \end{array} \right.}

\newcommand\gap{\varepsilon} 

\newcommand\observationset{{\mathcal Y}}

\def\P{{\mathbb P}}        
\def\E{{\mathbb E}}        
\def\1{{\mathbf 1}}        

\def\deq{\stackrel{\scriptscriptstyle d}{=}}         

\DeclareMathOperator{\parent}{pa}

\DeclareMathOperator*{\argmin}{argmin}

\newcommand{\ud}{\,\mathrm{d}}

\def\vertexset{{\mathscr V}}
\def\edgeset{{\mathscr E}}
\def\leavesset{{\mathscr L}}

\DeclareMathOperator*{\Poi}{Poi}
\DeclareMathOperator*{\Perm}{Perm}

\DeclareMathOperator*{\child}{child}

\def\mutkf{\mu_{\textrm{TKF}}}
\def\lambdatkf{\lambda_{\textrm{TKF}}}

\def\Sigmagap{\Sigma_{\varepsilon}}

\def\ppsample{{\mathbf X}}

\newtheorem{theorem}{Theorem} 
\newtheorem{lemma}[theorem]{Lemma} 
\newtheorem{proposition}[theorem]{Proposition}

\begin{document} 

\date{} 

\title{Evolutionary Inference via the Poisson Indel Process\footnote{Accepted for the Proceedings of the National Academy of Sciences.}}

\maketitle 

\vspace{-0.5in} 
\begin{center} 
\noindent {\normalsize \sc Alexandre Bouchard-C\^ot\'e$^1$, 
and Michael I. Jordan$^2$}\\ 
\noindent {\small \it 
$^1$Department of Statistics, University of British Columbia, V6T 1Z2, Canada;\\ 
$^2$Department of Statistics and Computer Science Division, 
University of California, Berkeley, 94720, USA}\\ 
\end{center} 

\bigskip 

\begin{abstract}
We address the problem of the joint statistical inference of phylogenetic trees 
and multiple sequence alignments from unaligned molecular sequences.  This problem 
is generally formulated in terms of string-valued evolutionary processes along 
the branches of a phylogenetic tree.  The classical evolutionary process, the 
TKF91 model~\cite{Thorne1991a}, is a continuous-time Markov chain model 
comprised of insertion, deletion and substitution events.  Unfortunately 
this model gives rise to an intractable computational problem---the 
computation of the marginal likelihood under the TKF91 model is exponential 
in the number of taxa~\cite{Miklos2003}.  In this work, we present a 
new stochastic process, the Poisson Indel Process (PIP), in which the 
complexity of this computation is reduced to linear.  The new model is 
closely related to the TKF91 model, differing only in its treatment of 
insertions, but the new model has a global characterization as a Poisson 
process on the phylogeny.  Standard results for Poisson processes allow 
key computations to be decoupled, which yields the favorable computational
profile of inference under the PIP model.  We present illustrative experiments
in which Bayesian inference under the PIP model is compared to separate
inference of phylogenies and alignments.
\end{abstract}

\medskip

\section{Introduction} 

The 
field of phylogenetic inference is being transformed by the
rapid growth in availability of molecular sequence data.  There is an urgent
need for inferential procedures that can cope with data from large numbers
of taxa and that can provide inferences for ancestral states and evolutionary 
parameters over increasingly large timespans.  Existing procedures are often 
not scalable along these dimensions and can be a bottleneck in analyses of 
modern molecular datasets.

A key issue that renders phylogenetic inference difficult is that sequence data 
are generally not aligned a priori, having undergone evolutionary processes that 
involve insertions and deletions.  Consider Figure~\ref{fig:evolutionary-tree},
which depicts an evolutionary tree in which each node is associated with
a string of nucleotides, and where the string evolves via insertion, deletion
and substitution processes along each branch of the tree.  Even if we 
consider evolutionary models that are stochastically independent along the 
branches of the tree (conditioning on ancestral states), the inferential problem 
of inferring evolutionary paths (conditioning on observed data at the leaves 
of the tree) does not generally decouple into independent computations
along the branches of the tree.  Rather, alignment decisions made throughout
the tree can influence the posterior distribution on alignments along any branch.

\begin{figure}[t] %
\begin{center}
\includegraphics[width=3in]{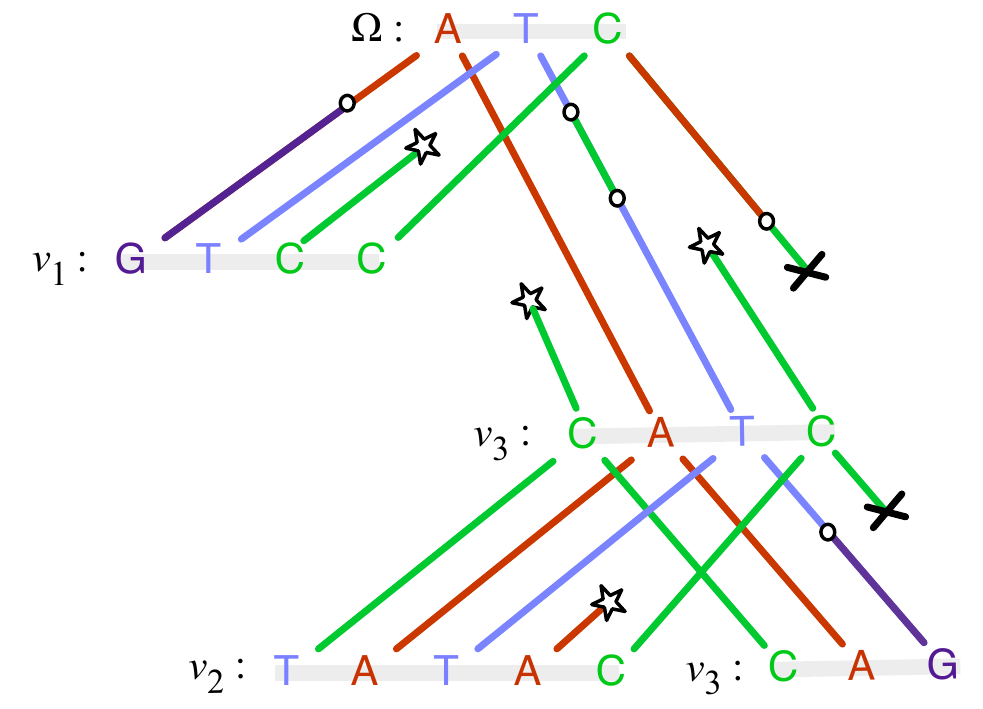} 
  \caption{A depiction of the evolution of a set of strings of nucleotides
     along the branches of a tree with leaves $\leavesset = \{v_1, v_2, v_3\}$ and root $\Omega$, where each string is subject to insertion,
     deletion and substitution processes.  Stars denote nucleotide insertion events, crosses, deletion events, and circles, substitution events.}
     \label{fig:evolutionary-tree}
\end{center}
\end{figure} 

This issue has come to the fore in a line of research beginning in 1991 with 
a seminal paper by Thorne, Kishino and Felsenstein~\cite{Thorne1991a}.
In the ``TKF91 model,'' a simple continuous-time Markov chain (CTMC) provides
a string-valued stochastic process along each branch of an evolutionary 
tree.  This makes it possible to define joint probabilities on trees and 
alignments, and thereby obtain likelihoods and posterior distributions for 
statistical inference.  A further important development has been the realization 
that the TKF91 model can be represented as a hidden Markov model (HMM), and 
that generalizations to a broader class of string-valued stochastic processes 
with finite-dimensional marginals are therefore possible~\cite{
Holmes2001,Hein2001TKF,Steel2001TKF,Metzler2001,Hein2003RecursionsPNAS,miklos2004,Lunter2005,Novak2008StatAlign}.
This has the appeal that statistical inference under these processes (known 
as transducers) can be based on dynamic programming~\cite{Allison1992HMM,
Krogh1994HMM,Searls1995,Mohri2009}.  Unfortunately, however, despite some 
analytic simplification that is feasible in restricted cases~\cite{Song2006Recurs}, 
the memory needed to represent the state space in these models is generally
exponential in the number of leaves in the tree~\cite{Dreyer2008}.  
Moreover, even in the simple TKF91 model, there does not appear to be 
additional structure in the state space that allows for simplification 
of the dynamic program.  Indeed, the running time of the most sophisticated 
algorithm for computing marginals \cite{Miklos2003} depends on the number of 
homology linearizations, which is exponential in sparse alignments~\cite{Schwartz2006}.

As a consequence of this unfavorable computational complexity, there has 
been extensive work on approximations, specifically on approximations to 
the joint marginal probability of a tree and an alignment, obtained by integrating 
over the derivation~\cite{Lunter2005,Westesson2012IndelRecon}.  A difficulty, however, is 
that these marginal probabilities play a role in tree inference procedures as 
the numerators and denominators of acceptance probabilities for Markov chain
Monte Carlo algorithms.  Loss of accuracy in these values can have large,
uncontrolled effects on the overall inference.  A second approach is to
consider joint models that are not  obtained by marginalization of a joint 
continuous-time string-valued process. A range of combinatorial \cite{Sankoff1975,Wheeler1994Malign,Lancia1999,Varon2010POY,Snir2011ParsimoniousIndelHistories,Loytynoja2012PRANK} and probabilistic \cite{Hein1990Unified,Knudsen2003LongApprox,Rivas2005,Redelings2005,Redelings2007Baliphy2} models fall in this category.
Although often inspired by continuous-time processes, obtaining a coherent and calibrated estimate of uncertainty in these models is difficult.

A third possible response to the computational complexity of joint inference 
of trees and alignments is to retreat to methods that treat these problems
separately.  In particular, as is often done in practice, one can obtain
a Multiple Sequence Alignment (MSA) via any method (often based on a 
heuristically chosen ``guide tree''), and then infer a tree based on
the fixed alignment.  This latter inferential process is generally based
on the assumption that the columns of the alignment are independent, in
which case the problem decouples into a simple recursion on the tree
(the ``Felsenstein'' or ``sum-product'' recursion~\cite{Felsenstein1981}).
Such an approach can introduce numerous artifacts, however, both in the
inferred phylogeny~\cite{Redelings2005,Redelings2007Baliphy2,Wong2008}, 
and in the inferred alignment~\cite{Roshan2006TreeAlignImprov,Nelesen2008}.

It is also possible to iterate the solution of the MSA problem and the tree 
inference problem~\cite{Liu2009a,Liu2012Sate}, which can be viewed as a heuristic
methodology for attempting to perform joint inference.  The drawbacks of 
these systems include a lack of theoretical understanding, the difficulty
of getting calibrated confidence intervals, and over-alignment problems
\cite{Schwartz2006,Lunter2004}.  

Finally, other methods have focused on analyzing only pairs of sequences at 
a time~\cite{Saitou1987NJ,Roch2010Matrices,Schwartz2006,Bradley2009FastStatAlign}.
While this approach can considerably simplify computation~\cite{Holmes2004RNAStruct,
Crawford2012BD}, it has the disadvantage that it is not based on an underlying
joint posterior probability distribution.

In the current paper we present a new approach to the joint probabilistic inference 
of trees and alignments.  Our approach is based on a model that is closely 
related to TKF91, altering only the insertion process while leaving the 
deletion and substitution processes intact.  Surprisingly, this relatively 
small change has a major mathematical consequence---under the new model 
evolutionary paths have an equivalent global description as a Poisson 
process on the phylogenetic tree.  We are then able to exploit standard 
results for Poisson processes (notably, Poisson thinning) to obtain 
significant computational leverage on the problem of computing the joint 
probability of a tree and an alignment.  Indeed, under the new model this 
computation decouples in such a way that this joint probability can be 
obtained in linear time (linear in the number of taxa and the sequence 
length), rather than in exponential time as in TKF91. 

Our new model has two descriptions: the first as a local continuous-time
Markov process that is closely related to the TKF91 model, and the second
as a global Poisson process.  We treat the latter as the fundamental description
and refer to the new process as the \emph{Poisson Indel Process} (PIP).  The
new description not only sheds light on computational issues, but it also 
opens up new ways to extend evolutionary models, allowing, for example, models 
that incorporate structural constraints and slipped-strand mispairing phenomena.

Under the Poisson process representation, another interesting perspective on our 
process is to view it as a string-valued counterpart to stochastic Dollo 
models~\cite{Alekseyenko2008,Nicholls2006}, which are defined on finite state spaces.
In particular, the general idea of the two steps generation process used in Section~\ref{sec:pscp-model} 
has antecedents in the literature on probabilistic modeling of morphological or lexical 
characters, but the literature did not address  the string-valued processes that are our focus here.

The remainder of the paper is organized as follows.  Section~\ref{sec:bg}
provides some basic background on the TKF91 model.  In Section~\ref{sec:pscp-model}
presents the PIP model, in both its local and global formulations.
Section~\ref{sec:pscp-comp-aspects} delves into the computational aspects 
of inference under the PIP model, describing the linear-time algorithm for 
computing the exact marginal density of an MSA and a tree.  In Section~\ref{sec:pscp-exp}
we present an empirical evaluation of the inference algorithm, and finally we
present our conclusions in Section~\ref{sec:pscp-discussion}.

\section{Background}\label{sec:prelim}
\label{sec:bg}

We begin by giving a brief overview of the TKF91 model.  Instead of following
the standard treatment based on differential equations, we present a Doob-Gillespie 
view of the model \cite{Doob1945DoobGillespieAlgo,Gillespie1977DoobGillespieAlgo}
that will be useful in our subsequent development.

Let us assume that at some point in time $t$, a sequence has length $n$.  
In the TKF91 process, the sequence stays unchanged for a random interval of 
time $\Delta t$, and after this interval, a single random mutation (substitution, 
insertion or deletion) alters the sequence.  This is achieved by defining a
total of $3n+1$ independent exponential random variables, $n$ of which 
correspond to deletion of a single character, $n$ of which correspond to 
the mutation of a single character and $n+1$ of which yield insertions after 
one of the $n$ characters (including one ``immortal'' position at the leftmost 
position in the string).  These $3n+1$ exponential random variables are 
simulated in parallel, and the value of the smallest of these random variables 
determines $\Delta t$.  The index of the winner determines the nature of the 
event at time $\Delta t$ (whether it is a substitution, deletion or insertion).

The random variables corresponding to a deletion have exponential rate $\mutkf$ 
while those corresponding to an insertion have exponential rate $\lambdatkf$.  
If the event is a mutation, a multinomial random variable with parameters 
obtained from the substitution rate matrix $\theta$ is drawn to determine the 
new value of the character.  Finally, if an insertion occurs, a multinomial 
random variable is drawn to determine the value of the new character, with 
parameters generally taken from the stationary distribution of $\theta$.  

This describes the evolution of a string of characters along a single edge
of a phylogenetic tree.  The extension to the entire phylogeny is straightforward;
we simply visit the tree in preorder and apply the single-edge process to 
each edge.  The distribution of the sequences at the root is generally assumed 
to be the stationary distribution of the single-edge process (conceptually, 
the distribution obtained along an infinitely-long edge).

Although the TKF91 model is reversible (and the PIP model as well, as we prove in Section~\ref{sec:properties}), making the location of the root 
unidentifiable, it is useful to assume for simplicity that an arbitrary root 
has been picked, and we will make that assumption throughout.  The likelihood 
is not affected by this arbitrary choice.

\section{The Poisson Indel Process}
\label{sec:pscp-model}

In this section we introduce the \emph{Poisson Indel Process} (PIP).
This process has two descriptions, a local description which is closely
related to the TKF91 model, and a global description as a Poisson process.

We require some additional notation (see Figure~\ref{fig:pip-notation}).
A phylogeny $\tau$ will be viewed as a continuous set of points, and its topology 
will be denoted by $(\vertexset, \edgeset)$, where $\vertexset \subset 
\tau$ is equal to the finite subset containing the branching points, the leaves  $\leavesset\subset \vertexset$ and the   
root $\Omega$, and where $\edgeset$ is the set of edges.  Parent nodes will be 
denoted by $\parent(v)$, for $ v\in\vertexset$, and the branch lengths by $b(v)$, 
which is the length of the edge from $\parent(v)$ to $v$.   
For any $x \in \tau$ (whether $x$ is a branch point in $\vertexset$, or 
an intermediate point on an edge), we write $\tau_x$ for the rooted phylogenetic 
subtree of $\tau$ rooted at $x$ (dropping all points in the original tree that are not descendants
of $x$).  Finally, the set of characters (nucleotides or amino acids) will be denoted $\Sigma$.

\begin{figure}[t] %
\begin{center}
\includegraphics[width=3in]{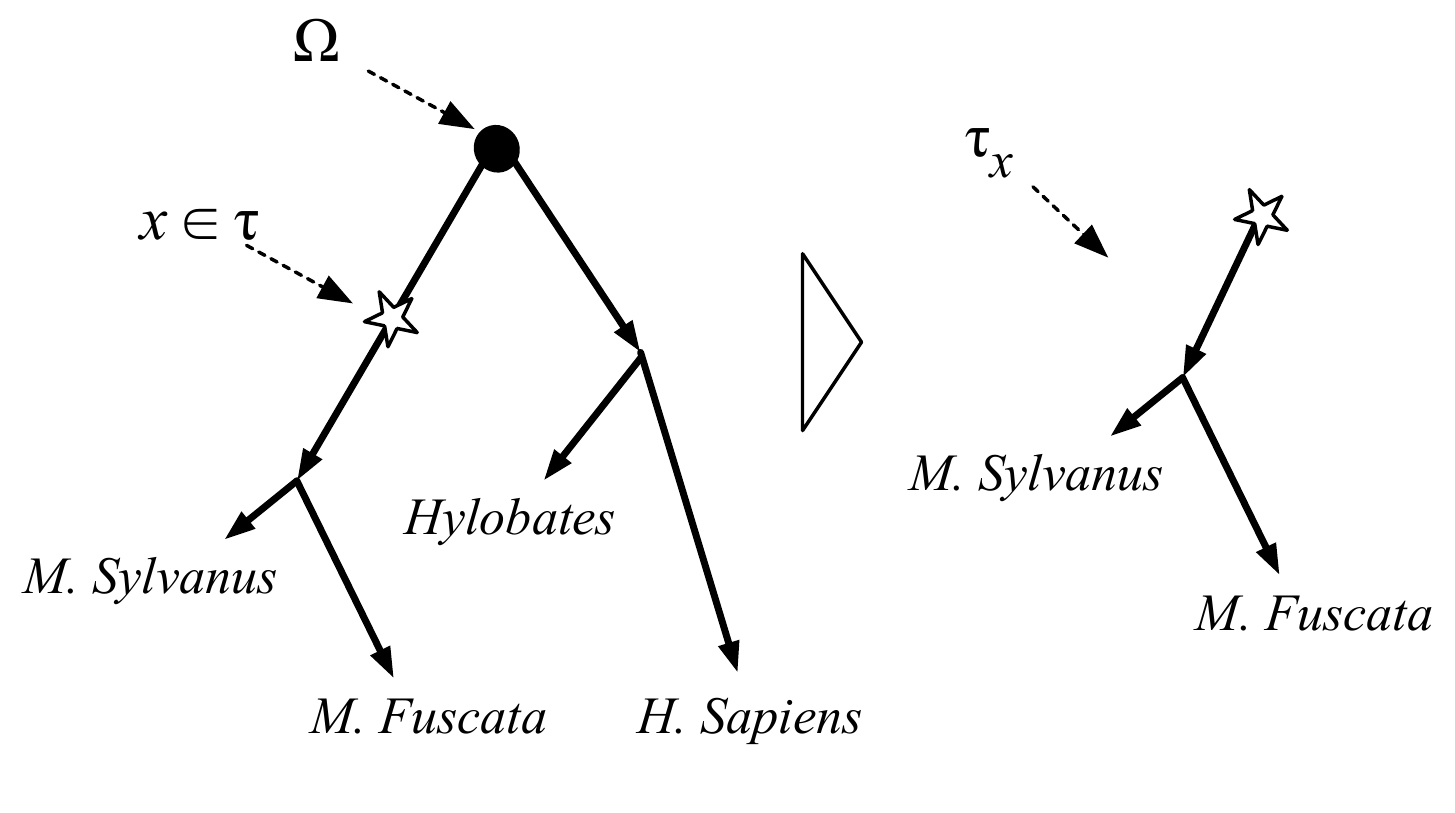} 
  \caption{Notation used for describing the PIP.  Given a phylogenetic tree 
  $\tau$ and a point $x\in\tau$ on that tree, $\tau_x$ is defined as the 
  subtree rooted at $x$.}
   \label{fig:pip-notation}
\end{center}
\end{figure}

\subsection{Local description}\label{sec:local-descr}

The stochastic process we propose has a local description that is very similar 
to the TKF91 process, the only change being that the insertion rate no longer
depends on the sequence length.  Therefore, instead of using $3n+1$ competing 
exponential random variables to determine the next event as in the TKF91 model 
($n$ for substitutions, $n+1$ for insertions, and $n$ for deletions), we now  
have $2n+1$ variables ($n$ for substitutions, $1$ for insertion, with rate 
$\lambda$, and $n$ for deletion, each of rate $\mu$).  When an insertion 
occurs, its position is selected uniformly at random.\footnote{More precisely, 
assume there is a real number $r_i$ in the interval $[0,1]$ assigned to each 
character in the string in increasing order: $0 < r_1 < r_2 < \dots < r_n < 1$.  
When an insertion occurs, sample a new real number $r'$ uniformly in the interval 
$[0,1]$ and insert the new character at the unique position (with probability one) 
such that an increasing sequence of real numbers $0  < \dots < r' <  \dots <  1$ 
is maintained.}  We assume that the process is initialized by sampling a 
Poisson-distributed number of characters, with parameter $\lambda/\mu$.  
Each character is sampled independently and identically according to the 
stationary distribution of $\theta$.

Note that if $\lambda / \lambdatkf$ is an integer, and the sequence has length $(\lambda / \lambdatkf) -1$ at some 
point in time, the distribution over the time and type of the next mutation 
is the same as in TKF91, by using the fact that the minimum of exponential 
variables with $\lambda_i$ is exponential, with rate equal to the sum of the 
$\lambda_i$.  However, in general the distributions are different.  We
discuss some of the biological aspects of these differences in 
Section~\ref{sec:pscp-discussion}; for now we focus on the computational
and statistical aspects of the PIP model.

\subsection{Poisson process representation}
\label{sec:poisson-descr}

We turn to a seemingly very different process for associating character 
strings with a phylogeny.  This process consists of two steps, the first 
involving insertions and the second involving deletions and substitutions. 

In the first step, depicted in Figure~\ref{fig:pip-gen}A, a multiset of 
insertion points is sampled from a Poisson process defined on the phylogeny $\tau$~\cite{Huelsenbeck1999Poisson}.  The rate measure for this Poisson process has atomic mass at the 
root of the tree; hence the need for multisets rather than simple point sets. 
Except for the root, no other points on the tree have an atomic mass (in particular, and in contrast to population genetics models, 
the probability that evolutionary events occur at branching points is zero). 
We denote this multiset of insertion points by $\ppsample$.

\begin{figure}[p] %
\begin{center}
\includegraphics[width=2.5in]{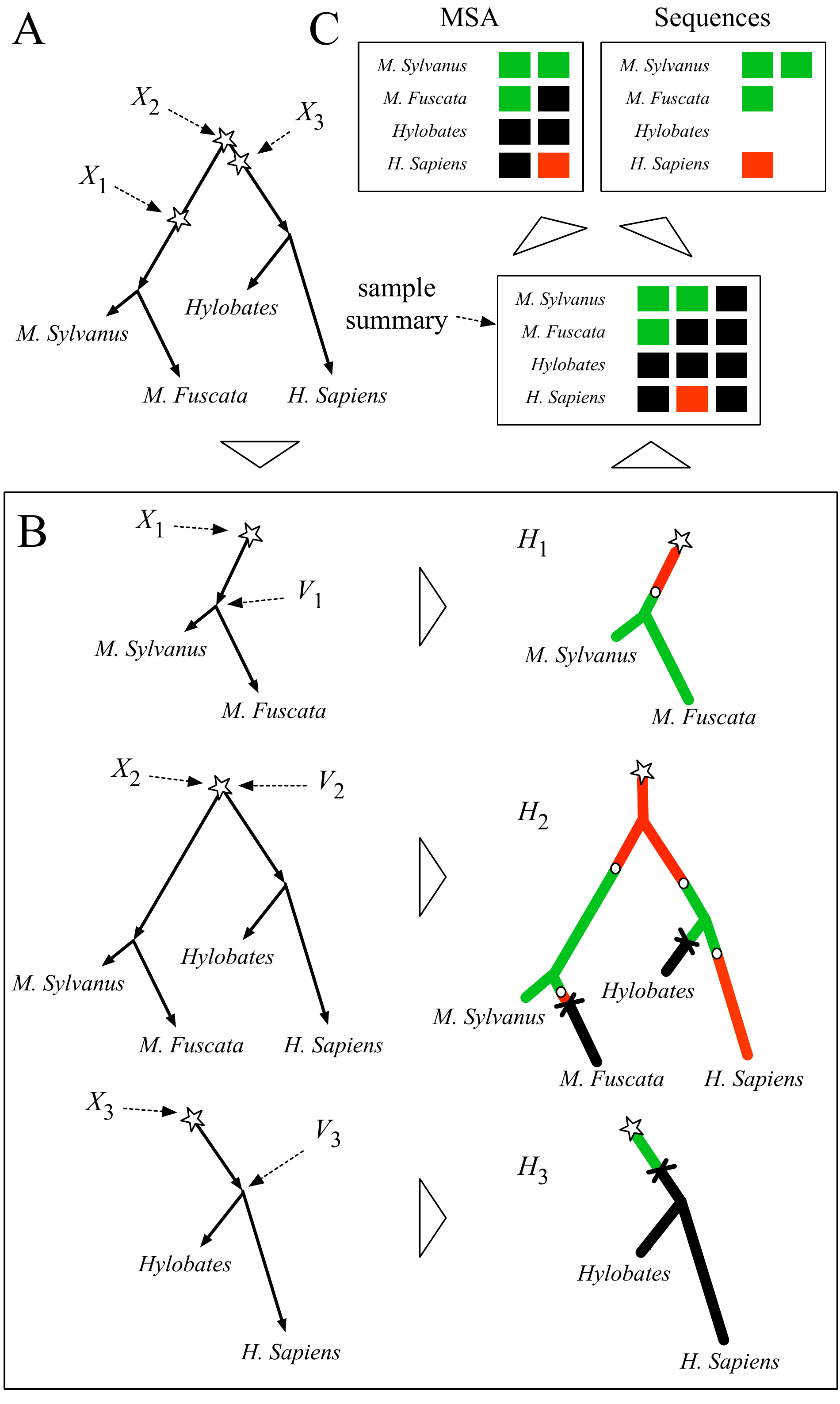} 
  \caption{Example of a PIP sample.  Here $\Sigma$ has two symbols, represented by red and green squares, and the absorbing deletion symbol $\gap$ is represented in black. (A) A sample from a Poisson process on $\tau$.  (B) Each sampled point corresponds to a rooted tree on which a CTMC path is sampled.  (C) The alignments and sequences are obtained as a deterministic function of the first two steps.}
  \label{fig:pip-gen}
\end{center}
\end{figure}

In the second step, we visit the insertion points one at a time.  The 
order of the visits of the insertions is sampled uniformly at random, 
$(X_1, X_2, \dots, X_I) \sim \Perm(\ppsample)$.  An insertion visit consists 
of two substeps.  First, we extract the directed subtree rooted at the 
insertion location $X_i$.  Examples of these subtrees are shown in 
Figure~\ref{fig:pip-gen}B, left.  Second, we simulate the fate of the inserted 
character along $\tau_{X_i}$.  This is done via a \emph{substitution-deletion} 
CTMC whose state space $\Sigmagap = \Sigma \cup \{\gap\}$ consists of the
basic alphabet $\Sigma$ augmented with an empty string symbol $\gap$.  As
shown in Figure~\ref{fig:pip-gen}B, right, the substitution-deletion
CTMC yields paths along subtrees in which a single character either mutates
or is deleted.  The latter event, represented by $\gap$, is an absorbing state.

We define a \emph{homology path} $H_i$ as the single-character history 
generated by a substitution-deletion CTMC along a phylogeny.  If a point 
$x \in \tau$ is a descendant of the insertion $X_i$, $H_i(x)$ is set to 
the state of the substitution-deletion CTMC at $x$.  If $x \in \tau$ is 
not a descendant of $X_i$, we set $H_i(x)$ to the absorbing symbol $\gap$.
Thus, formally, a homology path $H_i$ is a random map from any point on 
$\tau$ to $\Sigmagap$.  

Given a set of homology paths for each inserted character index $i$, 
the sequence at any point on the tree, $x\in \tau$, is obtained as follows 
(see Figure~\ref{fig:pip-gen}C, right).  First, we construct a list of all the 
values taken by $H_i(x)$ at the given point: $(H_1(x), H_2(x), \dots, H_I(x))$.  
Second, we remove from the list any characters that are equal to the absorbing 
symbol $\gap$.  The string obtained thereby is denoted by $Y(x)$.  The set
of observed data comprises the values of $Y$ at the leaves of the tree: 
$\observationset = \{(v, Y(v)) : v \in \leavesset\}$.

We can also construct an MSA $M$ from a set of homology paths (see 
Figure~\ref{fig:pip-gen}C, left).  From each homology path $H_i$, we
extract the characters at the leaves, arranging these characters in
a column.  Delete any column in which all of the characters are the
character $\gap$.  Arranging these columns in the order of the visits to
the insertion points.  The resulting matrix, whose entries range over
the augmented alphabet $\Sigmagap$, is the MSA $M$. 

For a given rooted phylogenetic tree $\tau$, we will denote by 
$p_\tau(m)$ the marginal probability that this process generates an 
MSA $m$, integrating over all homology paths, $p_\tau(m) = \P(M = m)$.
For joint inference, we make the phylogenetic tree $T$ random, with a distribution specified by a prior with density $p(\tau)$.

\subsection{Characterization}\label{sec:properties}

In this section we show that the local and the global descriptions of
the PIP given in the previous two subsections are in fact alternative 
descriptions of the same string-valued stochastic process.  In stating 
our theorem, we let $\nu$ denote the rate measure characterizing the 
insertion process in the global description, and let $Q$ and $\pi$ 
denote the transition matrix and the initial distribution for the 
substitution-deletion CTMC.
\begin{theorem}\label{prop:equivalence}
Let $\tau$ be a phylogenetic tree with an arbitrary rooting, and let us 
denote the Lebesgue measure on $\tau$ by the same symbol.  For any insertion 
rate $\lambda > 0$, deletion rate $\mu > 0$, and reversible substitution 
rate matrix $\theta$, the local and global processes described in 
Sections~\ref{sec:local-descr} and \ref{sec:poisson-descr} coincide 
if we set, for all $\sigma, \sigma'\in\Sigmagap$:
\begin{align*}
\nu(\ud x) &= \lambda \left(\tau(\ud x) + \frac{1}{\mu}\delta_\Omega(\ud x) \right), \\
Q_{\sigma,\sigma'} &= \bracearraycond{-\sum_{\sigma''\neq \sigma'} Q_{\sigma,\sigma''}
&\textrm{if }\sigma=\sigma' \\ 0&\textrm{if }\sigma=\gap \\ \mu&\textrm{if }\sigma'=\gap \\ 
\theta_{\sigma,\sigma'}&\textrm{o.w.,}}
\end{align*}
and set $\pi$ to be the quasi-stationary distribution of $Q$ \cite{Buiculescu1972Quasi}.  
\end{theorem}
The proof is given in the \sirefproofmainprop.
Note that in the case of interest here, where the rate of deletion does not 
depend on the character being deleted, $\pi_\sigma$ is equal to the entry 
of the stationary distribution of $\theta$ corresponding to $\sigma$ when 
$\sigma \neq \gap$, and zero otherwise.
The following result establishes some basic properties of the PIP model.  
Its proof can be found in the \sirefproofmainprop.
\begin{proposition}\label{prop:reversibility} 
For all $\mu, \lambda > 0$ and reversible rate matrix $\theta$, the PIP model 
is reversible, with a stationary length distribution given by a Poisson distribution 
with mean $\lambda/\mu$.
\end{proposition}

The Poisson stationary length distribution represents a modeling advantage of PIP over TKF91, which has a geometrically distributed stationary distribution. 
Based on a study of protein-length distributions for the three domains of life \cite{Zhang2000ProteinLenDist}, the Poisson distribution has been suggested \cite{Miklos2003PoiLen} as a more adequate length distribution.

From Proposition~\ref{prop:reversibility}, we can also obtain an 
alternative reparameterization of the PIP model, in terms of asymptotic expected 
length $\eta = \lambda/\mu$ and indel intensity $\zeta = \lambda \cdot \mu$.

\section{Computational Aspects}
\label{sec:pscp-comp-aspects}
\label{sec:pscp-marg-ll}

We turn to a consideration of the computational consequences of the 
Poisson representation of the PIP model.  We first consider how the Poisson 
process characterization allows us to compute the marginal likelihood, 
$p_\tau(m)$, in linear time; a significant improvement over methods
based on the TKF91 model.  We provide in \sirefcompaspects\ a brief discussion of the 
role that the marginal likelihood plays in inference.

To compute the marginal likelihood, $p_\tau(m)$, we first condition 
on the number of homology paths, $|\ppsample|$.  While the number of 
homology paths is random and unknown, we know that it can be no 
less than the number of columns $|m|$ in the postulated alignment $m$.  
We need to consider an unknown and unbounded number of birth events 
with no observed offsprings in the MSA, but as they are exchangeable, 
they can be marginalized analytically.  This is done as follows:
\begin{align*}
p_\tau(m) &= \E\big[ \P(M=m||\ppsample|) \big] \\
&= \sum_{n=|m|}^\infty \P(|\ppsample| = n) \cdot \binom{n}{|m|} \cdot 
(p(c_\emptyset))^{n-|m|} \prod_{c\in m} p(c),
\end{align*}
where the first factor captures the probability of sampling $n$ homology 
paths, the second, the number of ways to pick the $|m|$ observed homology paths (the columns, which contain at least one descendent character at the leaves) 
out of the $n$ paths, the factor $p(c) = \P(C=c)$ is the likelihood of 
a single MSA column $c$, and $c_\emptyset$ is a column with an absorbing 
deletion symbol at every leaf $v\in \leavesset$: $c_\emptyset \equiv \gap$ 
(in this section, we drop subscripts for column-specific random variables 
such as $C$, $H$ and $X$ since they are exchangeable). Note that such 
simplification is not possible in the TKF91 model, because the rate of insertion depends on the length of the internal sequences, and hence of the deletion events.

This expression can be simplified by introducing the function $\varphi$ 
defined as follows for all $z\in(0,1), k \in \{1, 2, \dots\}$:
\begin{align*} 
\varphi(z,k) &= \frac{1}{k!} \Vert \nu \Vert^k \exp\big\{ (z - 1) \Vert \nu \Vert \big\},\\
\Vert \nu \Vert &= \lambda\left(\Vert \tau \Vert + \frac{1}{\mu}\right),
\end{align*}
where $\Vert \tau \Vert$ is the normalization of the measure $\tau$, 
i.e., the sum of all the branch lengths in the topology.  We show in 
the \sirefproofcomp\ that this yields the simple formula:
\begin{align*}
p_\tau(m) = \varphi(p(c_\emptyset),|m|) \prod_{c\in m} p(c).
\end{align*}

The next step is to compute the likelihood $p(c)$ of each individual 
alignment column $c$.  We do this by partitioning the computation into subcases depending on the   
location of the tree at which the insertion point $X$ is located for column $c$.  
More precisely, we look at the most recent common ancestor $V=v \in \vertexset$ 
of the characters in $c$ that are not equal to $\gap$ (see Figure~\ref{fig:pip-gen}B).  If $v \neq \Omega$, 
this corresponds to the most recent endpoint of the edge $e \in \edgeset$ 
where the insertion occurred.

Computing the prior probability of the insertion location is  
greatly simplified by the fact that $X | |\ppsample| \sim \bar \nu$ 
(see \cite{Kingman1993}, Chapter~2.4), where $\bar \nu = \nu / \Vert \nu \Vert$ 
denotes the probability obtained by normalizing the measure $\nu$.  We can therefore write:
\begin{align*} 
\P(V = v) &= \bracearraycond{\bar\nu(e\backslash\{\Omega\}) & 
\textrm{if }v\neq\Omega\\\bar \nu(\{\Omega\}) &\textrm{o.w.}} \\
& = \frac{1}{\Vert \tau \Vert + 1/\mu} \times \bracearraycond{b(v) &
\textrm{if }v\neq\Omega\\1/\mu &\textrm{o.w.}}
\end{align*}
Finally, the column probabilities are computed as follows:
\begin{align*}
\P(C = c) &= \sum_{v\in \vertexset} \P(V = v) \P(C = c | V = v) \\
&= \sum_{v\in \vertexset} \P(V = v) f_v,
\end{align*}
where $f_v$ is the output of a slight modification of Felsenstein's 
peeling recursion \cite{Felsenstein1981} applied on the subtree rooted 
at $v$ (the derivation for $f_v$ can be found in the \sirefproofcomp).
Since computing the peeling recursion for one column takes time 
$O(|\leavesset|)$, we get a total running time of $O(|\leavesset|\cdot |m|)$, 
where $|\leavesset|$ is the number of observed taxa, and $|m|$ is the 
number of columns in the alignment.

\section{Experiments}
\label{sec:pscp-exp}

We implemented a system based on our model that performs joint 
Bayesian inference of phylogenies and alignments.  We used this 
system to quantify the relative benefits of joint inference relative to 
separate inference under the PIP and TKF91 models; i.e., the benefits 
of inferring trees on accuracy of the inferred MSA and the benefits of 
inferring MSAs on the accuracy of the inferred tree.

We used synthetic data to assess the quality of the tree reconstructions 
produced by PIP, compared to the reconstructions of PhyML 2.4.4, a 
widely-used platform for phylogenetic tree inference \cite{Guindon2004}.  
We also compared the inferred MSAs to those produced by Clustal 2.0.12  
\cite{Higgins1988}, a popular MSA inference system. 

While our implementation evaluated in this section is based on the Bayesian framework, we evaluate it using a frequentist methodology. More precisely, we use Bayes estimators (described in \sirefcompaspects) to obtain two point estimates from the posterior, one for the MSA, and one for the phylogeny. Each point estimate is compared to the true alignment and tree. It is therefore possible to compare the method to the well-known frequentist methods mentioned above. 

In this study, we explored four types of potential improvements: (1) resampling trees and MSAs increasing the quality of inferred MSAs, 
	compared to resampling only MSAs; (2) resampling trees and MSAs increasing the quality of inferred trees, 
	compared to resampling only trees; (3) resampling trees increasing the quality of inferred trees, 
	compared to trees inferred by PhyML, and fixing the MSA to the one 
	produced by Clustal; (4) resampling MSAs increasing the quality of inferred MSAs, compared 
	to MSAs inferred by Clustal, and fixing the tree to the one produced 
	by PhyML.
The results are shown in \resultstable. 
These experiments were 
based on 100 replicas, each having 7 taxa at the leaves, a topology sampled 
from the uniform distribution, branch lengths sampled from rate 2 exponential 
distributions, indels generated from the PIP with parameters $\eta=100, \zeta=1$,  
and nucleotides sampled from the Kimura two-parameter model (K2P)~\cite{kimura1980}.

\begin{table}[t] 
\caption{PIP results on simulated data} 
\begin{center} 
\begin{tabular}{llcccc} 
\toprule 
\multirow{2}{*}{\rotatebox{90}{{\footnotesize Exp.}}} & Tree resampled? & \nolabel & \yeslabel & \nolabel & \yeslabel \\ 
& MSA resampled? & \nolabel & \nolabel & \yeslabel & \yeslabel \\ 
\midrule 
\midrule 
\multirow{3}{*}{\rotatebox{90}{MSAs}} & Edge recall (SP) 
& {\bf 0.25} & - & 0.22 & 0.24 \\ 
& Edge Precision & 0.22 & - & 0.56 & {\bf 0.58} \\ 
& Edge F1 & 0.23 & - & 0.31 & {\bf 0.32} \\ 
\midrule 
\multirow{2}{*}{\rotatebox{90}{Trees}} & Partition Metric & 0.24 & 0.22 & - & {\bf 0.19} \\ 
& Robinson-Foulds & 0.45 & 0.38 & - & {\bf 0.33} \\ 
\bottomrule 
\end{tabular} 
\label{table:pscp-synt} 
\end{center} 
\end{table} 

We measured the quality of MSA reconstructions using the F1 score, defined as the harmonic mean of the reconstructed alignment edge recall (called the sum-of-pairs score or developer's score in the MSA literature \cite{Sauder2000MSAMetrics}) and alignment edge precision (modeler's score \cite{Zachariah2005MSAMetrics}). We measured the quality of tree reconstructions using the partition (symmetric clade difference) metric \cite{Bourque1978} and the weighted Robinson-Foulds metric \cite{Robinson1979Metric}. Relative improvements were obtained by computing the absolute value of the quality difference (in terms of the F1 for alignments, and Robinson-Foulds distance for trees), divided by the initial value of the measure. We report relative improvements averaged over the 100 replicas.

We observed improvements of all four types.  Comparing Edge F1 relative 
improvements to Robinson-Foulds relative improvements, the relative additional 
improvement of type (2) is larger (13\%) than that of type (1) (3\%).  
Overall (i.e., comparing the baselines to the joint system), the full 
improvements of both trees and MSAs are substantial: 43\% Edge F1 improvement, 
and 27\% Robinson-Foulds improvement.
See Figure~\ref{fig:relative-improvements} for a summary of the relative 
improvements. 

\begin{figure}[t] %
\begin{center}
\includegraphics[width=3in]{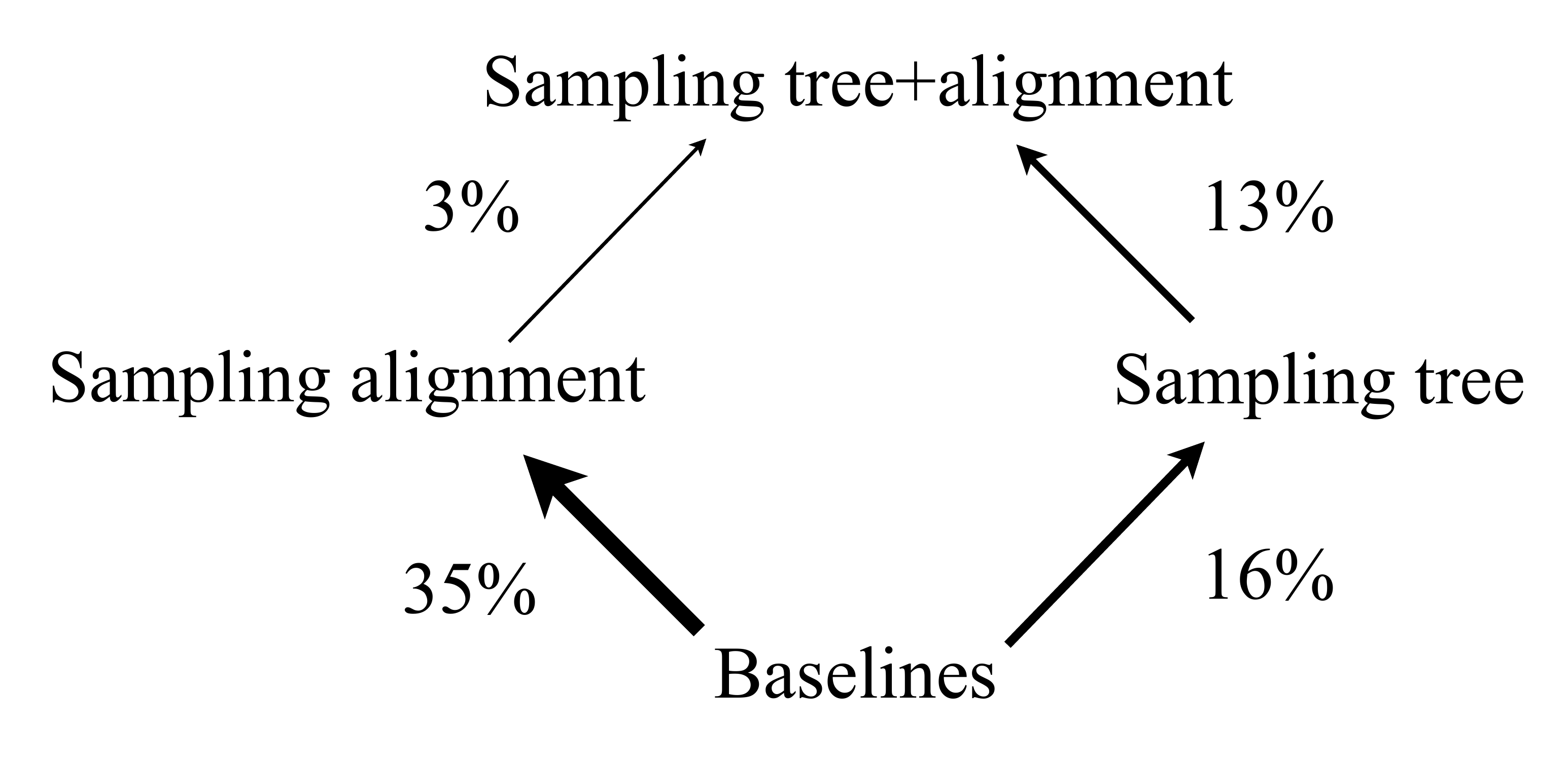} 
  \caption{Relative improvements for enabling each component of the sampler.  Arrows on the left are relative alignment improvements, and arrows on the right are relative tree improvements.  }
  \label{fig:relative-improvements}
\end{center}
\end{figure}

We also tested our system on data generated from the TKF91 model instead of 
the PIP model.  We used the same tree distribution and number of replicas as 
in the previous experiments, and the same generating TKF91 parameters as 
\cite{Holmes2001}.
We again observed improvements over the baselines, both in terms of MSA 
and tree quality.  For MSAs, the relative improvement over the baseline 
was actually larger on the TKF91-generated data than on the PIP-generated 
data (47\% versus 43\%, as measured by Edge F1 improvement over Clustal), 
and lower but still substantial for phylogenetic trees (13\% versus 27\%, as measured by 
Robinson-Foulds improvement over PhyML).

It should be noted that the MCMC kernels used in these experiments (described in the \sirefproposals) are based on simple Metropolis-Hastings proposals, and can therefore suffer from high rejection rates in large datasets. Fortunately, previous work in the statistical alignment literature has developed sophisticated  MCMC kernels, some of which could be applied to inference in our model, for example \cite{Redelings2005}. Another potential direction would be to replace MCMC by a Sequential Monte Carlo posterior approximation \cite{Bouchard2012Phylogenetic}.

It should also be emphasized that point indels is certainly not the exclusive driving force behind sequence evolution. In particular, ``long indels'' (atomic insertions and deletions of long segments, with a probability higher than the product of their point indels), are also prominent. As a consequence, any system purely based on point indels will have significant biases on biological data. In practice, these biases will introduce three undesirable artifacts: overestimation of the branch lengths; ``gappy alignments,'' where the reconstructed MSA has many scattered gaps instead of a few long ones; and the related ``ragged end'' problem, where the prefix and suffix of sequences are poorly aligned because observed sequences are often truncated in practice.  In the next section, we propose ways to address these limitations.

\section{Discussion}
\label{sec:pscp-discussion}

We have presented a novel string-valued evolutionary model that can be used
for joint inference of phylogenies and multiple sequence alignments.  As with
its predecessor, the TKF91 model, the new model can be used to capture the 
homology of characters evolving on a phylogenetic tree under insertion, deletion 
and substitution events.  Its advantage over TKF91 is that it 
permits a representation as a Poisson process on the tree.  This representation
has the consequence that the marginal likelihood of a tree and an alignment
(marginalizing over ancestral states) can be computed in time linear in the
number of taxa, rather than exponential as in the case of TKF91. 
Poisson representations have played an important role in pure substitution processes 
\cite{Huelsenbeck2000Poisson,Nicholls2006,Alekseyenko2008}, but to the best of 
knowledge, this is the first time Poisson representations are used for indel inference.

Although the insertion process in TKF91 might be argued to be more realistic
biologically than that of the PIP model, in that it allows the insertion rate
to vary as the sequence length varies, in the common setting in which all of 
the sequences being aligned are of roughly similar lengths, this extra degree 
of freedom may be of limited value for inference.  Indeed, in our experiments 
we saw that the PIP model can perform well even when data are generated from 
the TKF91 model.  We might also note that there are biological processes in 
which insertions originate from a source that is extrinsic to the sequence 
(e.g., viruses or other genomic regions), in which case the constant-rate
assumption of PIP may actually be preferred.

It is also important to acknowledge, however, that neither TKF91 nor PIP are
accurate representations of biology.  Their use in phylogenetic modeling reflects
the hope that the statistical inferences they permit---most notably taking into 
account the effect of indels on the tree topology---will nonetheless be useful 
as data accrue.  This hope is more likely to be realized in larger datasets, 
motivating our goal of obtaining a method that scales to larger sets of species.  
But both models should also be viewed as jumping-off points for further modeling 
that is more faithful to the biology while retaining the inferential power of 
the basic models.  For example, there has been significant work on extending TKF91 
to models that capture the ``long indels'' that arise biologically but are not 
captured by the basic model~\cite{Thorne1992,Miklos2001,miklos2004}.  

In this regard, we wish to note that the Poisson representation of the PIP 
model provides new avenues for extension that are not available within 
the TKF91 framework.  In particular, the superposition property of Poisson 
processes makes it possible to combine the PIP model with other models
that follow a Poisson law.  For example, if the location $X'$ of long indels, 
slipped-strand mispairing~\cite{Kelchner2000ManualAlign} or other non-local 
changes follow a Poisson point process, the union $U=X \cup X'$ of the non-local 
changes with the point indels $X$ provided by a PIP will also be distributed 
according to a  Poisson process.  Moreover, the thinning property of Poisson 
processes provides a principled approach to inference for such superpositions.   
Indeed, an MCMC sampler for the superposition model can be constructed as 
follows: first, we can exploit the decomposition to analytically marginalize  
$X$ (using the algorithm presented in this paper).  Second, the other terms 
of the superposition and the sequences at these point in time can be represented 
explicitly as  auxiliary variables.  Since we have an efficient algorithm 
for computing the marginal likelihood, the auxiliary variables can be resampled 
easily.  Note that designing an irreducible sampler without marginalizing $X$ 
would be difficult: integrating out $X$ creates a bridge of positive probability 
between any pair of patterns of non-local changes.

Under the parameterization of the process used in this paper, the model assumes 
both an equal deletion rate for all characters, and a uniform probability over 
inserted characters. It is worth noting that our inference algorithm can be 
modified to handle models relaxing both assumptions, by replacing the 
calculation of $\beta(v)$ in \sirefproofcomp\ by a quasi-stationary 
distribution calculation~\cite{Buiculescu1972Quasi}. It would be interesting 
to use this idea to investigate what non-uniformities are present in biological indel data.

Finally, another avenue to improve PIP models is to make the insertion 
rate mean measure more realistic: instead of being uniform across the tree, 
it could be modeled using a prior distribution, hence forming a Cox 
process \cite{Cox1955}.  This would be most useful when the sequences 
under study have large length or indel intensity variations across sites and branches~\cite{Sniret2006IndelProfiling}.

\subsection*{Acknowledgments}

We would like to thank Bastien Boussau, Ian Holmes, Michael Newton and Marc Suchard for their comments and suggestions. 
This work was partially supported by a grant from the Office of Naval 
Research under Contract Number N00014-11-1-0688, by grant K22 HG00056  
from the National Institutes of Health, and by grant SciDAC BER KP110201 
from the Department of Energy. 

\bibliographystyle{plain} 

\bibliographystyle{pnas} 
\bibliography{references}

\newcommand\refequivalence{Theorem~1}
\newcommand\refreversibility{Proposition~2}

\appendix
 

\section{Proofs for the Main PIP Properties}
\label{ap:equiv}

In this section, we prove \refequivalence\ and \refreversibility.  
We begin by stating and proving two lemmas.

\begin{lemma}
\label{lemma:exponential}
Let $U \sim \Unifd(0,t)$ and $W \sim \Expd(\mu)$ be independent for 
fixed $t, \mu > 0$.  Then
\begin{align*} 
\P(W + U > t) = \frac{1-\exp(-t \mu)}{t \mu}.
\end{align*}
\end{lemma}
\begin{proof}
By conditioning:
\begin{align*} 
\P(W + U > t) &= \E\big[\P(W + U > t | U) \big] \\
&= \int_0^t \frac{\exp(-x\mu) }{t} \ud x \\
&= \frac{1-\exp(-t \mu)}{t \mu}.
\end{align*}
\end{proof}

\begin{lemma}
\label{lemma:stat}
Let $\tau_0$ denote a degenerate topology consisting of a root $\Omega$ 
connected to a single leaf $v_0$ by an edge of length $t$.  Let $H_i$ be a 
homology path as defined in the main paper, with $\tau = \tau_0$.  
For all $x\in \tau$, define $I(x) = \{i : H_i(x) \neq \gap, 1 \le i \le I\}$ and:
\begin{align*} 
N &= |I(\Omega)| \\
N' &= |I(v_0)|.
\end{align*}
Then $N\sim \Poi(\lambda/\mu)$ implies $N' \sim \Poi(\lambda/\mu)$.
\end{lemma}

\begin{proof} 
To prove the result, we decompose $N$ and $N'$ as follows
(see Figure~\ref{fig:reversibility}):
\begin{align*}
N_1 &= |I(\Omega) \backslash I(v_0)| \\
N_2 &= |I(\Omega) \cap I(v_0)| \\
N_3 &= |I(v_0) \backslash I(\Omega)| \\
N_4 &= |I \backslash I(\Omega) \backslash I(v_0)| \\
N &= N_1 + N_2 \\
N' &= N_2 + N_3.
\end{align*}
By the Coloring Theorem \cite{Kingman1993},
\begin{align*} 
N_2 \sim \Poi\left( \nu(\{\Omega\}) \P(W > t) \right),
\end{align*}
where $W$ is a rate $\mu$ exponential random variable, and $\nu$ is as 
in the condition of \refequivalence.  Therefore 
$N_2 \sim \Poi(\lambda \exp(-t\mu) / \mu)$.  Similarly, 
\begin{align*} 
N_3 \sim \Poi\left(  \nu(\tau \backslash \{\Omega\}) \P(W + U > t) \right),
\end{align*}
where $U \sim \Unifd(0,t)$, and therefore from Lemma~\ref{lemma:exponential}, 
$N_3 \sim \Poi(\lambda (1-\exp(-t\mu)) / \mu)$.  It follows that:
\begin{align*} 
N' &= N_2 + N_3 \\
&\sim \Poi\left(\frac{\lambda}{\mu} e^{-\mu} + 
\frac{\lambda}{\mu}\left(1 - e^{-\mu}\right)\right) \\
&= \Poi\left(\frac{\lambda}{\mu}\right),
\end{align*}
which concludes the proof of the lemma.
\end{proof}

\begin{figure}[t]  %
\begin{center}
\includegraphics[width=2in]{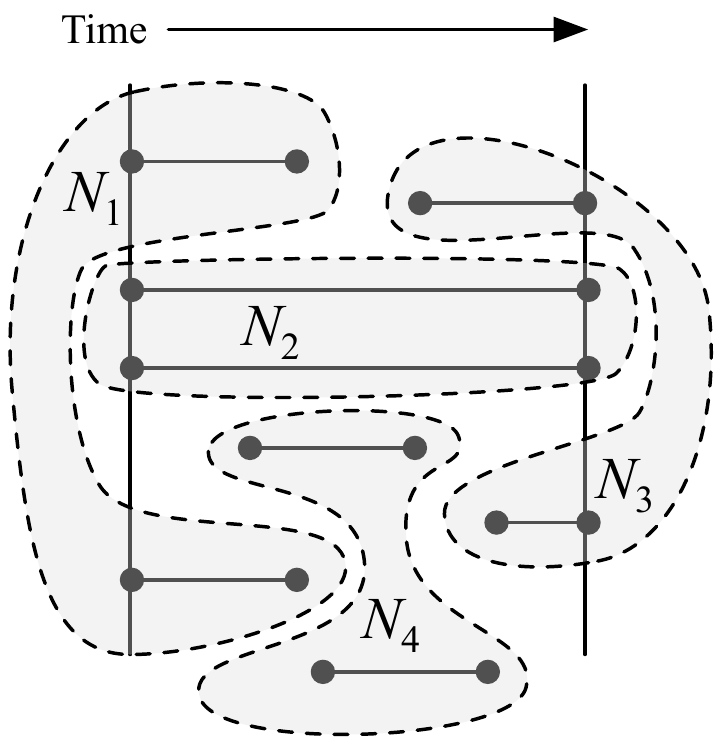} 
  \caption{Notation used in the appendix.  The horizontal lines denote the times 
   where each character is present in the sequence.  The vertical line on the left 
   denotes the sequence at $\Omega$, and the vertical line on the right, the sequence 
   at $v_0$.  The sites are decomposed depending on whether they are present at each 
   of two points $\Omega,v_0$ in  $\tau_0$.}
  \label{fig:reversibility}
\end{center}
\end{figure}

We can now prove \refequivalence: \\
\begin{proof} 
In order to establish the equivalence, it is enough to show that for all edges 
$e = (v \to v')$ in the tree, the following two properties hold:
\begin{enumerate} 
\item The distribution of the string length at the ancestral endpoint, 
$|Y(v)|$, is identical in the local and global descriptions: a Poisson 
distribution with rate $\lambda/\mu$.
\item The distribution of the number and locations of mutations that 
fall on $e\backslash\{v, v'\}$ are also identical in the local and global 
descriptions.
\end{enumerate}
We will enumerate the edges in the tree in preorder, using induction to
establish these two hypotheses on this list of edges.

In the base case, hypothesis~1 is satisfied by construction: the local 
description is initialized with a $\Poi(\lambda/\mu)$-distributed number 
of characters, and in the global description, the intensity measure $\nu$ 
of the Poisson process $\ppsample$ assigns a point mass $\lambda/\mu$ 
to $v = \Omega$.

To establish hypothesis~1 in the inductive case, let $e' = (v'' \to v)$ 
denote the parent edge.  By hypothesis~1 on $e'$, $|Y(v'')|\sim\Poid(\lambda/\mu)$, 
therefore by Lemma~\ref{lemma:stat} and hypothesis~2 on $e'$, hypothesis~1 
is satisfied on $e$ as well.

To establish hypothesis~2, it is enough to show that for all 
$x\in e\backslash\{v, v'\}$ the waiting time for each type of mutation 
given $Y(x)$ is exponential, with rates:
\begin{enumerate}[(a)]
  \item $\lambda$ for insertion,
  \item $\mu\cdot|Y(x)|$ for deletion, and
  \item $\sum_{\sigma\neq \gap} \theta_{\sigma,\sigma'} |Y(x)|_\sigma$ 
  for substitutions to $\sigma'\neq \gap$, where $|s|_\sigma$ denotes 
  the number of characters of type $\sigma \in \Sigma$ in the string $s\in\Sigma^*$.
\end{enumerate}
Item~(a) follows from the Poisson Interval Theorem \cite{Kingman1993}.  
Items~(b) and (c) follow from the standard Doob-Gillespie characterization 
of CTMCs: if $X_t$ is a CTMC with rate matrix $Q = (q_{i,j})$ and $Z_{i,j}$ 
are independent exponential random variables with rate $q_{i,j}$, then 
$$(\Delta, J) | (X_0 = i) \deq (\min_{j\neq i} Z_{i,j}, \argmin_{j\neq i} Z_{i,j}),$$ 
where $\Delta = \inf \{t : X_t \neq i\}$, $J = X_{\Delta}$.
\end{proof}

We now turn to \refreversibility\ and establish reversibility. \\
\begin{proof}
Let $h(n_1, n_2, n_3, n_4) = \P(N_i = n_i, i\in\{1,2,3,4\})$.  
Using reversibility of $\theta$, it is enough to show that $h$ is 
invariant under the permutation $(1\ 3)$; i.e., 
$h(n_1, n_2, n_3, n_4) = h(n_3, n_2, n_1, n_4)$.

We have that $h(n_1, n_2, n_3, n_4)$ is equal to:
{\footnotesize 
\begin{align*}
\P\Big(N_i = n_i, &\sum_i N_i = \sum_i n_i, N_1 + N_2 = n_1 + n_2, N_3+N_4 = n_3+n_4\Big)   \\
=&\  \P\Big(\sum_i N_i = \sum_i n_i\Big) \times \\
&\ \P\Big(N_1 + N_2 = n_1 + n_2, N_3+N_4 = n_3+n_4\Big|\sum_i N_i = \sum_i n_i\Big) \times \\
&\ \P(N_1 = n_1, N_2 = n_2 | N_1 + N_2 = n_1 + n_2) \times \\
&\ \P(N_3= n_3, N_4 = n_4 | N_3+N_4 = n_3+n_4) \\
=&\ f_1(n_1 + n_2 + n_3 + n_4) \times \\
&\ \left(\frac{1/\mu}{1/\mu + t}\right)^{n_1 + n_2} 
\left(\frac{t}{1/\mu+t}\right)^{n_3+n_4} \times \\
&\ \left(1 - e^{-\mu t}\right)^{n_1} f_2(n_2) \times \\
&\ \left(\frac{1 - e^{-\mu t}}{t \mu}\right)^{n_3} f_3(n_4),
\end{align*}}
where  only the dependencies of the functions $f_1, f_2$ and $f_3$ is 
 important in this argument, not their exact form.  By inspection, it is clear that 
$h$ is invariant under the permutation $(1\ 3)$.
\end{proof}

\section{Proofs for the Likelihood Computation}

First, we show how the function $\varphi$, defined in 
the main paper, simplifies the computation 
of $p_\tau(m)$:
\begin{align*}
p_\tau(m) &= \E\big[ \P(M=m||\ppsample|) \big] \\
&= \sum_{n=|m|}^\infty \P(|\ppsample| = n) \cdot \binom{n}{|m|} \cdot 
(p(c_\emptyset))^{n-|m|} \prod_{c\in m} p(c) \\
&= \frac{e^{\Vert \nu \Vert} \prod_{c\in m} p(c)}{|m|! (p(c_\emptyset))^{|m|}} 
\sum_{n = |m|}^\infty \frac{\left(\Vert \nu \Vert p(c_\emptyset) \right)^n}{(n - |m|)!} \\
&=\frac{e^{\Vert \nu \Vert} \left(\Vert \nu \Vert p(c_\emptyset) \right)^{|m|}  
\prod_{c\in m} p(c)}{|m|! (p(c_\emptyset))^{|m|}} \sum_{k = 0}^\infty 
\frac{\left(\Vert \nu \Vert p(c_\emptyset) \right)^k}{k!} \\
&= \frac{e^{\Vert \nu \Vert} \left(\Vert \nu \Vert p(c_\emptyset) \right)^{|m|}  
\prod_{c\in m} p(c)}{|m|! (p(c_\emptyset))^{|m|}} 
\exp\left(\Vert \nu \Vert p(c_\emptyset) \right) \\
&= \varphi(p(c_\emptyset),|m|) \prod_{c\in m} p(c).
\end{align*}

Next, we show how to compute $f_v = \P(C = c|V=v)$ for all $v\in \vertexset$. The recursions for $f_v$ are similar to those found in stochastic Dollo models \cite{Alekseyenko2008}.
Note first that $f_v$ can be zero for some vertices.  To see where and why, 
consider the subset of leaves $S$ that that have an extant nucleotide  
 in the current column $c$, $S = \{v\in\leavesset : H(v)\neq\gap\}$.  Then $f_v$ will be non-zero only 
for the vertices ancestral to all the leaves in $S$.  Let us call this set of 
vertices $A$ (see Figure~\ref{fig:pscp-fel}).

\begin{figure}[t] %
\begin{center}
\includegraphics[width=2in]{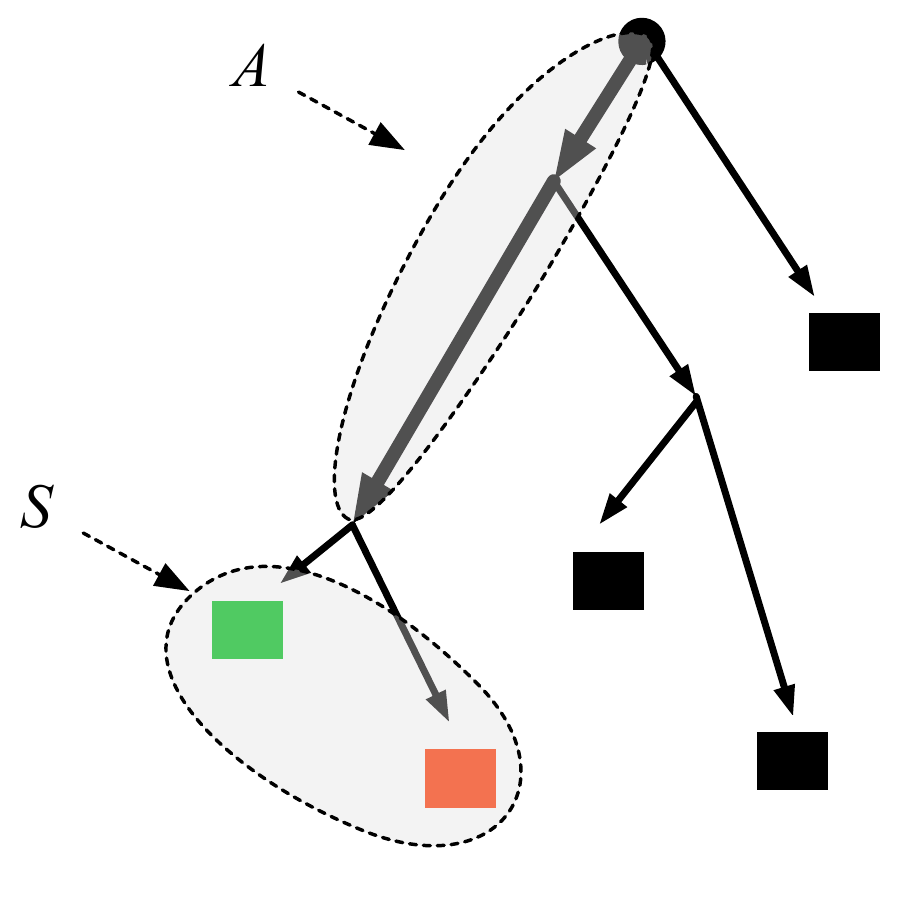} 
  \caption{Given a set $S$ of leaves $v$ with $H(v) \neq \gap$, we define the set 
  $A$ of vertices with nonzero modified Felsenstein peeling weight to be those ancestral 
  to the leaves in $S$. In this example, $A$ contains three vertices.}
  \label{fig:pscp-fel}
\end{center}
\end{figure}

To compute $f_v$ on the remaining vertices, we introduce an intermediate
variable, $\tilde f_v = \P(C = c|V=v, H(v)\neq \gap)$.  This variable
can be computed using the standard Felsenstein peeling recursion (dynamic
programming) as follows:
{\small 
\begin{align}
\tilde f_v(\sigma) &= \bracearraycond{\1(c(v) = \sigma) &
\textrm{if }v\in \leavesset \\ \sum_{\sigma'\in\Sigmagap} 
\exp(b(v)Q)_{\sigma,\sigma'} \prod_{w\in\child(v)} \tilde f_{w}(\sigma')&\textrm{o.w.}} \label{eq:tfvs} \\
\tilde f_v &= \sum_{\sigma\in\Sigma} \pi_\sigma \tilde f_v(\sigma). \label{eq:tfv}
\end{align}}

From Lemma~\ref{lemma:exponential}, we have an expression for the survival 
probability at $v$ given an insertion on the edge $(\parent(v)\to v)$:
\begin{align} 
\beta(v) &= \P(H(v) \neq \gap|V=v) \notag \\
&= \frac{1}{b(v)} \frac{1}{\mu} \left( 1 - e^{-\mu b(v)} \right). \label{eq:betas}
\end{align}

Finally, for $c\neq c_\emptyset$, we have:
\begin{align}
f_v &= \P(C = c|V=v) \notag \\
&= \E[ \P(C = c|V=v, H(v)) ] \notag \\
&= \label{eq:fv} \bracearraycond{\tilde f_v&\textrm{if }v=\Omega\\ 
\1[v\in A] \beta(v) \tilde f_v &\textrm{o.w.,}}
\end{align}
and for $c = c_\emptyset$:
\begin{align}
f_v &= \bracearraycond{\tilde f_v&\textrm{if }v=\Omega\\ 1 + \beta(v) 
(\tilde f_v - 1) &\textrm{o.w.}} \label{eq:fvgap}
\end{align}

\section{Proposal distributions}
\label{sec:pscp-proposals}

To perform full joint inference over trees and alignments using Markov
chain Monte Carlo, several objects need to be resampled: the tree 
topology, the branch lengths, the MSA, and the parameters.

For trees and branch lengths, we use standard proposal mechanisms \cite{Lakner2008}.
Our MSA proposal is inspired by the proposal of \cite{Lunter2005}, avoiding the mixing problems of auxiliary variables~\cite{Holmes2001,Jensen2002,Bouchard-Cote2009b}. 
Our proposal distribution consists of two steps.  
First, we partition the leaves into two sets $A, B$.
Given a current MSA $m_0$, the support of the proposal is the set $S$ of MSAs $m$ satisfying the following constraints:
\begin{enumerate}
  \item If $e$ has both endpoints in $A$ (or both in $B$), then 
  $e\in m \Longleftrightarrow e\in m_0$.
  \item If $e,e'$ have both endpoints in $A$ (or both in $B$), then 
  $e \prec_m e' \Longleftrightarrow e \prec_{m_0} e'$.
\end{enumerate}
The notation $\prec_m$ is based on the concept of posets over the columns 
(and edges) of an MSA \cite{Schwartz2006}.   

We propose an element $m^*\in S$ with probability proportional to 
$\prod_{c\in m^*} p(c)$.  The set $S$ has exponential size, but can be 
sampled efficiently using standard pairwise alignment dynamic programming.  
A Metropolis-Hastings ratio is then computed to correct for $\varphi$.  
Note that the proposal induces an irreducible chain: one possible outcome 
of the move is to remove all links between two groups of sequences.  
The chain can therefore move to the empty MSA and then construct any MSA 
incrementally.

For the parameters, we used multiplicative proposals in the $(\lambda,\mu)$  
parameterization \cite{Lakner2008}.

\section{Computational Aspects}

In this section, we provide a brief discussion of the 
role that the marginal likelihood plays in both frequentist and Bayesian inference methods.

\subsection{Maximum likelihood}

In the case of maximum likelihood, the overall inference problem 
involves optimizing over the marginal likelihood:
\begin{align*}
\sup_{\tau\in\phylospace(\leavesset),m\in\alignspace(y)} \log p_\tau(m),
\end{align*}
where $\tau$ ranges over phylogenies on the leaves $\leavesset$, and $m$ 
ranges over the alignments consistent with the observed sequences $y$.
This optimization problem can be approached using simulated annealing, 
where a candidate phylogeny and MSA pair $(\tau',m')$ is proposed at each 
step $i$, and is accepted (meaning that it replaces the previous candidate 
$(\tau,m)$) according to a sequence of acceptance functions $f^{(i)}(p,p')$ 
depending only on the marginal probabilities $p = p_\tau(m), p' = p_{\tau'}(m')$.  
Provided $\lim_{i\to\infty} f^{(i)}(p,p') = \1[p' > p]$ sufficiently slowly, 
this algorithm converges to the maximum likelihood phylogeny and 
MSA \cite{Delyon1988Annealing}.

\subsection{Bayes estimators}

In order to define a Bayes estimator, one typically specifies a decision 
space $D$ (for example the space of MSAs, or the space of multifurcating 
tree topologies, or both), a projection into this space, $(\tau,m) \mapsto 
\rho(\tau,m)\in D$, and a loss function $l:D\to [0,\infty)$ on $D$ (for example, 
for tree topologies, the symmetric clade difference, or partition metric 
\cite{Bourque1978}; and for alignments, $1-$ the edge recall or Sum-of-Pairs 
(SP) score~\cite{Robert2001}).

Given these objects, the optimal decision in the Bayesian framework 
(also known as the consensus tree or alignment), is obtained by minimizing over $d\in D$ the risk 
$
\E[l(d,\rho(T,M))|\observationset].
$
This expectation is intractable, so it is 
usually approximated with the empirical distribution of the output 
$(\tau^{(i)}, m^{(i)})$ of an Markov chain Monte Carlo (MCMC) algorithm.
Producing MCMC samples boils down to computing acceptance ratios of the form:
$$\frac{p(\tau') p_{\tau'}(m')}{p(\tau)p_{\tau}(m)} \cdot 
\frac{q_{(\tau',m')}(\tau,m)}{q_{(\tau,m)}(\tau',m')},
$$
for some proposal having density $q$ with respect to a shared reference 
measure on $\phylospace(\leavesset)\times\alignspace(y)$.
We thus see that for both maximum likelihood and joint Bayesian inference
of the MSA and phylogeny the key problem is that of computing the marginal
likelihood $p_\tau(m)$. 

\section{Pseudocode and Example}
 
In this section, we summarize the likelihood computation. We also give a concrete numerical example to illustrate the calculation.
 
\begin{enumerate}
\item Inputs:
\begin{enumerate}
\item PIP parameter values $(\lambda, \mu)$, substitution matrix $\theta$ over $\Sigma$. \\ \emph{Example: $(\lambda, \mu) = (2.0,1.0), \Sigma = \{\textrm{a}\}$}
\item Rooted phylogenetic tree $\tau$ \\ \emph{Example: $\tau = ((v_2:1.0,v_3:1.0)v_0:1.0,v_4:2.0)v_1;$}
\item Multiple sequence alignment $m$ \\ \emph{Example:} $m = $\begin{verbatim}
v_2|-a
v_3|aa
v_4|a-
\end{verbatim}
\end{enumerate}
\item Computing modified Felsenstein recursion:
\begin{enumerate}
\item For each site, compute $\tilde f_v(\sigma)$ in post-order using Equation~(\ref{eq:tfvs}), and from each $\tilde f_v(\sigma)$, compute $\tilde f_v$ using Equation~(\ref{eq:tfv}) \\ \emph{Example: \\for site 1, $(\tilde f_{v_2},\tilde f_{v_3},\tilde f_{v_0},\tilde f_{v_4},\tilde f_{v_1}) = (0.0,1.0,0.23,1.0,0.012)$; \\for site 2, $(\tilde f_{v_2},\tilde f_{v_3},\tilde f_{v_0},\tilde f_{v_4},\tilde f_{v_1}) = (1.0,1.0,0.14,0.0,0.043)$; }
\item Do the same for an artificial site or column $c_{\emptyset}$ where all leaves have a gap \\ \emph{Example: \\for site 3, $(\tilde f_{v_2},\tilde f_{v_3},\tilde f_{v_0},\tilde f_{v_4},\tilde f_{v_1}) = (0.0,0.0,0.40,0.0,0.67)$; }
\end{enumerate}
\item For each node $v$ in the tree, compute the survival probability $\beta(v)$ using Equation~(\ref{eq:betas}) (setting it to 1 at the root for convenience) \\ \emph{Example: \\ $(\beta(v_2),\beta(v_3),\beta(v_0),\beta(v_4),\beta(v_1)) = (0.63,0.63,0.63,0.43,1.0)$}
\item For each site, compute the set of nodes $A$ ancestral to all extant characters, as described in the caption of Figure~\ref{fig:pscp-fel} \\\emph{Example: \\for site 1, $A = \{v_1\}$\\for site 2, $A = \{v_0,v_1\}$}
\item Computing $f_v$:
\begin{enumerate}
\item For each site, compute $f_v$ using Equation~(\ref{eq:fv}) \\\emph{Example: \\for site 1, $(f_{v_2},f_{v_3},f_{v_0},f_{v_4},f_{v_1}) = (0.0,0.0,0.0,0.0,0.012)$; \\for site 2, $(f_{v_2},f_{v_3},f_{v_0},f_{v_4},f_{v_1}) = (0.0,0.0,0.086,0.0,0.043)$; }
\item For $c_{\emptyset}$, use Equation~(\ref{eq:fvgap}) \\ \emph{Example: \\for site 3, $(f_{v_2},f_{v_3},f_{v_0},f_{v_4},f_{v_1}) = (0.37,0.37,0.62,0.57,0.67)$; }
\end{enumerate}
\item For each node $v$ in the tree, compute $\iota_v = \P(V = v)$ as shown in Section~3 of the main paper \\ \emph{Example: \\ $(\iota(v_2),\iota(v_3),\iota(v_0),\iota(v_4),\iota(v_1)) = (0.17,0.17,0.17,0.33,0.17)$}
\item Compute $p_\tau(m)$ from the $\iota_v$'s, $f_v$'s as shown in Section~3 of the main paper \\ \emph{Example: $\log p_\tau(m) = -11$}
\end{enumerate}

\end{document}